\documentclass[a4paper,11pt,reqno]{amsart}

\usepackage{amsmath,amssymb,amsthm, dsfont}
\usepackage{enumitem}
\usepackage{longtable}
\usepackage[colorlinks,citecolor=blue]{hyperref}
\usepackage[english]{babel}
\usepackage{csquotes}

\numberwithin{equation}{section}

\newtheorem{theorem}{Theorem}[section]

\newtheorem{proposition}[theorem]{Proposition}

\theoremstyle{definition}
\newtheorem{definition}[theorem]{Definition}
\newtheorem{remark}[theorem]{Remark}
\newtheorem{exercise}{Exercise}[section]

\newtheorem{example}[theorem]{Example}

\newcommand{\bra}[1]{\left( #1 \right)}
\newcommand{\dbra}[1]{\left \langle #1 \right |}
\newcommand{\dket}[1]{\left| #1 \right \rangle}
\newcommand{\sprod}[2]{\left< #1 | #2 \right>}
\newcommand{\sqa}[1]{\left[ #1 \right]}
\newcommand{\cur}[1]{\left\{ #1 \right\}}
\newcommand{\ang}[1]{\left< #1 \right>}
\newcommand{\abs}[1]{\left| #1 \right|}
\newcommand{\nor}[1]{\left\| #1 \right\|}

\newcommand{\cO}{\mathcal{O}}
\newcommand{\cY}{\mathcal{Y}}
\newcommand{\cH}{\mathcal{H}}
\newcommand{\cX}{\mathcal{X}}
\newcommand{\I}{\mathds{1}}
\newcommand{\cS}{\mathcal{S}}
\newcommand{\cop}{\mathcal{C}}
\newcommand{\cK}{\mathcal{K}}
\newcommand{\cL}{\mathcal{L}}
\newcommand{\R}{\mathbb{R}}
\newcommand{\C}{\mathbb{C}}
\newcommand{\tr}{\operatorname{tr}}

\usepackage[hyperref,doi,url=false,doi=false, isbn=false, giveninits=true, backref,style=numeric,maxbibnames=99, backend=biber]{biblatex}
\bibliography{Budapest-notes-trevisan-depalma.bib}

\title[QOT: Quantum Channels and Qubits]{Quantum Optimal Transport: Quantum Channels and Qubits
}

\author[G. De Palma]{Giacomo De Palma}
\address{G. DP.: Dipartimento di Matematica, Università di Bologna, 40126, Bologna, Italy }
\email{giacomo.depalma@unibo.it}

\author[D. Trevisan]{Dario Trevisan}
\address{D.T.: Dipartimento di Matematica, Università di Pisa, 56125 Pisa, Italy  }
\email{dario.trevisan@unipi.it}

\date{\today}

\begin{document}

\begin{abstract}
These notes are based on the lectures given by the second author at the School on Optimal Transport on Quantum Structures at Erd\"os Center in September 2022. The focus of the exposition is on two recently introduced approaches on quantum optimal transport: one based on quantum channels as generalized transport plans, the other based on the notion of Hamming-Wasserstein distance of order 1 on multiple-qubit systems. The material is presented in an elementary manner with a focus on the finite-dimensional setting.
\end{abstract}

\maketitle

\section{Introduction}

Quantum Optimal Transport  is a rapidly growing field at the intersection of quantum mechanics and optimal transport theory. While optimal transport theory searches for the most efficient way to transport resources or information from one location to another, in the quantum setting such a problem becomes more challenging due to the non-commutativity of the involved quantities, i.e., states and observables. In recent years, there has been a surge of interest in quantum optimal transport, from both theoretical and computational perspectives, making it an exciting and promising area of research.

These lecture notes closely follow the material presented by the second author in a series of lectures on quantum optimal transport given at the Erdős Center, Budapest, in September 2022. The main focus  is on two recent different proposals for defining quantum optimal transport and their properties. The first proposal, introduced in \cite{de2021quantum-channel}, is based on the use of quantum channels to transport one quantum state onto another. The second proposal, following \cite{de2021quantum-1}, is well-suited for composite systems, in particular multiple-qubit systems, and is based on a natural notion of neighbouring states. Both approaches provide distinct and complementary perspectives on the problem of quantum optimal transport, and we believe they provide significant examples and applications in the field. The notes aim to a detailed and accessible introduction to these proposals, in particular focusing on the mathematical foundations, and highlight some key concepts and some elementary applications, e.g.\ to concentration inequalities.

The exposition is structured as follow: in Section~\ref{sec:classical}, we review some basic notions on the classical optimal transport problem and the related Wasserstein distances. Section~\ref{sec:systems} collects some elementary facts on quantum systems, states, observables and quantum channels. Section~\ref{sec:overview} reviews some literature on quantum optimal transport, providing a rough classification of the main approaches to the subject. Section~\ref{sec:channels} is devoted to the formulation of quantum optimal transport using quantum channels, while Section~\ref{sec:qubits} focuses on the quantum Wasserstein distance on qubit systems.

Although the field is rapidly evolving, we hope that the readers may gain from these notes a basic mathematical understanding of some of the latest developments in quantum optimal transport and their potential impact on the field of quantum information science.

\subsection*{Acknowledgments} We thank the organizers of the School and Workshop on Optimal Transport on Quantum Structures at Erd\"os Center, J.\ Maas, S.\ Rademacher, T.\  Titkos, D.\ Virosztek, as well as the other lecturers, A.\ Figalli, E.A.\ Carlen and F.\ Golse, as well as the speakers and participants, for many stimulating conversations. We also acknowledge three anonymous referees for many suggestions and improvements in the manuscript.

GDP has been supported by the HPC Italian National Centre for HPC, Big Data and Quantum Computing - Proposal code CN00000013, CUP J33C22001170001 and by the Italian Extended Partnership PE01 - FAIR Future Artificial Intelligence Research - Proposal code PE00000013, CUP J33C22002830006.
GDP is a member of the ``Gruppo Nazionale per la Fisica Matematica (GNFM)'' of the ``Istituto Nazionale di Alta Matematica ``Francesco Severi'' (INdAM)''.

DT has been supported by the HPC Italian National Centre for HPC, Big Data and Quantum Computing - Proposal code CN1 CN00000013, CUP I53C22000690001, and by the INdAM-GNAMPA project 2023 ``Teoremi Limite per Dinamiche di Discesa Gradiente Stocastica: Convergenza e Generalizzazione''.

\section{Classical Optimal Transport}\label{sec:classical}

In this section, we review some notation and basic facts about the classical optimal transport theory. We provide no proofs and focus on the elementary setting of finite sets and  hence discrete measures: the general theory is exposed in several excellent monographs, e.g.\ \cite{villani2009optimal, ambrosio2005gradient, santambrogio2015optimal, figalli2021invitation, peyre2019computational}.

\subsection{Monge's problem}
Historically, G.\ Monge's seminal work on optimal transport theory, published in 1781, laid the foundations for a mathematical framework to study the optimal ways of transporting goods or mass from one location to another. He introduced the idea of finding an optimal transportation map as a way to minimize the total cost, or energy, required to move mass between two locations, and it provided a first mathematical framework to study this problem.

Although Monge considered only absolutely continuous distributions of mass (i.e., densities with respect to Lebesgue measure), to keep technicalities at a minimum, we focus instead on the following discrete formulation of the optimal transport problem. Given
 \begin{enumerate}
 \item finite sets $\cX$, $\cY$, representing the source and target positions of the masses (without loss of generality, one usually assumes $\cX = \cY$),
 \item a source distribution of mass $\sigma =  (\sigma(x))_{x\in \cX}$ (with $\sigma(x) \ge 0$ for every $x \in \cX$)
 \item a target distribution $\rho =  (\rho(y))_{y\in \cY}$, (with $\rho(y) \ge 0$ for every $y \in \cY$)
\item  and a cost function for transporting a unit of mass from position $x \in \cX$ to position $y \in \cY$,
\begin{equation} c: \cX \times \cY \to \R, \quad (x,y) \mapsto c(x,y),\end{equation}
\end{enumerate}
Monge's optimal transport problem searches for a transport map $T: \cX \to \cY$ that moves $\sigma$ into $\rho$ with minimal total transport cost, defined as
\begin{equation} \label{eq:monge} \sum_{x \in \cX} c(x, T(x)) \sigma(x).\end{equation}
The condition that $T$ moves $\sigma$ into $\rho$ can be stated as
\begin{equation} \label{eq:monge-map}\sum_{x \in T^{-1}(y)} \sigma(x) = \rho(y), \quad \text{for every $y \in \cY$.}\end{equation}
Summation upon $y \in \cY$, yields immediately that  the total masses of the two distributions must be equal, hence, due to homogeneity of \eqref{eq:monge}, one usually considers  only probability distributions.

%\begin{figure}[h]
%
%prova prova
%\includegraphics[width=0.5\textwidth]{MMP N=200bip=True.pdf}
%%\label{ot}
%%\caption{A random instance of OT in the plane}
%\end{figure}

\subsection{The Kantorovich problem}

Despite being a  natural formulation of the problem, simple examples show that transport maps $T$ satisfying \eqref{eq:monge-map} may not exist for general probability distributions $\sigma$, $\rho$. To overcome this issue, L.\ Kantorovich extended Monge's original approach by considering the possibility of splitting mass into fractions and directing it to different target locations. He introduced the concept of a coupling as a probability distribution over the product space and formulated the problem of finding an optimal coupling as a linear programming problem.

Precisely, a coupling $\pi$ between $\sigma$ and $\rho$ is  defined as a probability distribution on the product set $\cX \times \cY$, such that
\begin{equation} \sum_{x \in \cX} \pi(x,y) = \rho(y), \quad \sum_{y \in \cY} \pi(x,y) = \sigma(x), \quad\text{for every $x \in \cX$, $y \in \cY$,}\end{equation}
i.e., the marginal distributions are respectively $\sigma$ and $\rho$. Let us denote with $\cop(\sigma, \rho)$  the set of couplings between $\sigma$ and $\rho$, which is always non empty, since it contains the product distribution $\pi(x,y) = \sigma(x) \rho(y)$. The transport cost associated to a coupling $\pi$ is defined as the expectation
\begin{equation} \ang{c}_\pi = \sum_{x \in \cX , y \in \cY} c(x,y) \pi(x,y).\end{equation}
Kantorovich formulation extends Monge's problem, since any transport map $T: \cX \to \cY$ that moves $\sigma$ into $\rho$ in the sense of \eqref{eq:monge-map}, is naturally associated to the coupling $\pi(x,y) = \sigma(x) 1_{\cur{y=T(x)}}$, and $\ang{c}_\pi$ equals \eqref{eq:monge}.

Moreover, minimization over all couplings, keeping $\sigma$ and $\rho$ fixed, is a linear programming problem, i.e.\ on optimization problem of a linear cost functional with linear constraints in the variables, that can be solved e.g.\ via the simplex algorithm.

\subsection{The Wasserstein distance}

When $\cX = \cY$ and the cost $c(x,y) = d(x,y)$ is a distance, the optimal transport cost
\begin{equation}\label{eq:w1}  W_1(\sigma,\rho) = \min_{\pi \in \cop(\sigma, \rho)} \sum_{x,y \in \cX} d(x,y) \pi(x, y)\end{equation}
induces a distance between probability distributions over $\cX$, usually called Wasserstein distance, but also known as the Earth Mover's distance, as it measures the minimum amount of work needed to transform one probability distribution into another, where the amount of work is proportional to the distance travelled by each unit of mass during transportation.  More precisely, $W_1$ is called Wasserstein distance of order $1$, since for every $p\ge 1$, one can define the Wasserstein distance of order $p$ as
\begin{equation}  W_p(\sigma,\rho) = \min_{\pi \in \cop(\sigma, \rho)} \bra{ \sum_{x,y \in \cX} d(x,y)^p \pi(x, y)}^{1/p},\end{equation}
which also induces a distance. For $p \in (0,1)$, one defines instead
\begin{equation} W_p(\sigma, \rho) = \min_{\pi \in \cop(\sigma, \rho)}  \sum_{x,y \in \cX} d(x,y)^p \pi(x, y),\end{equation}
which is a Wasserstein distance of order $1$, with respect to the distance $(x,y) \mapsto d(x,y)^p$.

Kantorovich also introduced the fundamental concept of duality in linear programming problems. It is a general tool that establishes a relationship between two different linear programming problems, usually referred to as the primal and the dual one. Roughly speaking, the dual problem is obtained by taking the transpose of the matrix of coefficients defining the primal problem, and exchanging the roles of variables and constraints. In the case of the Wasserstein distance of order $1$, with respect to a distance $d$, its expression is particularly simple:
\begin{equation}\label{eq:kantorovich-classical} W_1(\sigma, \rho) = \max \cur{ \sum_{x \in \cX} f(x) (\sigma(x) - \rho(x) )\, : \, \abs{ f(x) - f(y)} \le d(x,y) \, \forall x, y},\end{equation}
i.e., we maximize the difference between the expectations $\ang{f}_\sigma- \ang{f}_\rho$ among all functions $f$ that are $1$-Lipschitz with respect to the distance $d$.

\subsection{Comparison with other distances}

Of course, there are plenty of other distances between probability distributions, such as
\begin{enumerate}[label = \alph*)]
\item  the Total Variation  distance
\begin{equation}\label{eq:tv}  \| \rho - \sigma\|_{TV} = \frac 1 2 \sum_{x \in \cX} | \sigma(x) - \rho(x)|\end{equation}
\item the Hellinger distance
\begin{equation}\label{eq:hellinger} H(\sigma, \rho) = \sqrt{ \frac 1 2 \sum_{x \in \cX}| \sqrt{\sigma(x)} - \sqrt{\rho(x)}|^2 }\end{equation}% =\sqrt{ 1 - \sum_{x \in \cX} \sqrt{\sigma(x)}\sqrt{\rho(x)}},\end{equation}
\item the relative entropy or  Kullback-Leibler divergence
\begin{equation}\label{eq:kl} D_{KL}(\sigma|| \rho) = \sum_{x \in \cX} \sigma(x) \ln \left( \sigma(x)/\rho(x) \right),\end{equation}
provided that $\sigma(x) = 0$ whenever $\rho(x) = 0$, otherwise we set it to $+\infty$  (although this is not in general an actual distance, it is commonly employed to compare two distributions).
\end{enumerate}
Comparing the definitions above with \eqref{eq:w1}, it is clear that a possible advantage of the Wasserstein distance is to make use of the underlying geometry on the set $\cX$ induced by the distance $d$.  This is indeed the case, and the Wasserstein distance has found numerous applications in many fields. If we limit ourselves to statistics and machine learning, one of its main applications is in quantifying the difference between empirical distributions of sampled points, which has immediate implications in data analysis. In particular, the Wasserstein distance has been proposed as a discriminator in generative models \cite{arjovsky2017wasserstein}, where it measures the distance between the generated and real distributions.  Additionally, the Wasserstein distance can be used as a tool for geometric interpolation between probabilities, allowing for natural transformation of one distribution into another \cite{peyre2019computational}.  The Wasserstein distance is also a theoretical tool that can be used to prove concentration of measure and other functional inequalities \cite{gozlan2010transport}, which has important implications in probability theory and statistics. The versatility and usefulness of the Wasserstein distance have made it a central concept in modern mathematics and its applications continue to be explored in various research areas.

While the Wasserstein distance has many useful applications, it also has some downsides compared to other metrics such as total variation or relative entropy. One of the main drawbacks becomes apparent in high-dimensional settings where the curse of dimensionality can become an issue when comparing empirical distributions of data -- but to be fair, this issue is common also for other distances. In addition to this, the computation of the Wasserstein distance involves solving an optimization problem, which can be computationally expensive. However, several solutions have been proposed to address these issues. For example, strictly convex terms can be added to the cost \cite{cuturi2013sinkhorn}, to improve the convergence of iterative algorithms, or in the dual formulation, the notion of Lipschitz functions can be relaxed, for instance, by parametrizing over a class of neural networks \cite{arjovsky2017wasserstein}. These solutions can help overcome the computational challenges associated with the Wasserstein distance and make it more practical to use in a variety of settings.

% Despite its limitations, the Wasserstein distance remains a valuable tool in modern mathematics and its applications continue to be explored in various fields.

\begin{exercise}\label{ex:TV-W}
Show that the Total variation distance is the Wasserstein distance with respect to the trivial distance $d(x,y) = 1_{\cur{x \neq y}}$. Deduce that
\begin{equation} \lim_{p \to 0^+} W_{p}(\sigma, \rho) = \| \rho -\sigma\|_{TV}.\end{equation}
\end{exercise}

\section{Quantum systems}\label{sec:systems}

Quantum mechanics is a mathematical framework that describes the behaviour of particles at the atomic and subatomic level. One of the key features of the theory is that it replaces commutative objects, such as functions and probabilities, with non-commutative ones, precisely given as operators on a complex Hilbert space.  To keep the exposition simple and avoid topological considerations, in this notes we focus only on the case of finite dimensional systems, which should be considered as the analogue of the case of finite sets and discrete measures from the previous section. Again, we give no proofs of the basic facts recalled in this section, and we recommend to  the interested reader any of the excellent monographs \cite{alicki2001quantum, nielsen2002quantum, bengtsson2017geometry, holevo2019quantum, moretti2019fundamental} for a detailed exposition.

\subsection{Systems, states and observables}

Any quantum system has an associated Hilbert space $\cH$ with a scalar product $\sprod{\cdot}{\cdot}$, which is conventionally anti-linear in the left argument.
With a small abuse of notation, in these notes we will identify the quantum system with its Hilbert space and denote them with the same symbol.
Following Dirac's notation, we write $\dket{\psi} \in \cH$  and $\dbra{\psi} \in \cH^*$ for the linear functional
\begin{equation} \cH \ni \dket{\varphi} \quad \mapsto \quad \sprod{\psi}{\varphi}.\end{equation}
Many classical objects of probability and measure theory have their natural quantum counterpart. The correspondence is summarized in Table \ref{table:dictionary} below. The most relevant ones are:
\begin{enumerate}
\item quantum (real-valued) observables, which correspond to classical functions (or random variables) and are given by linear self-adjoint operators $A: \cH \to \cH$, i.e., $\sprod{\psi}{A \varphi} = \sprod{A \psi}{\varphi}$ for every $\dket{\varphi}, \dket{\psi} \in \cH$,
\item quantum states, which correspond to classical probability distributions, given  by operators $\rho: \cH \to \cH$ that are self-adjoint, positive in the sense of quadratic forms, i.e., $\sprod{\psi}{\rho \psi} \ge 0$ for every $\dket{\psi} \in \cH$, and with unit trace  $\tr[\rho] = 1$. Such operators are also called density operators.
\end{enumerate}
We write $\cL(\cH)$ for the set of linear operators from $\cH$ into itself, $\cO(\cH)$ for the set of observables, and  $\cS(\cH)$ for the set of density operators. Notice that both $\cS(\cH) \subseteq \cO(\cH) \subseteq \cL(\cH)$. We write $\I_{\cH} \in \cO(\cH)$ for the identity operator (often simply $\I$ when the space $\cH$ is understood) and $\I_V$ for the orthogonal projector on a subspace $V < \cH$. Given a state $\rho \in \cS(\cH)$ and an observable $A \in \cO(\cH)$, its expected value is defined as
\begin{equation} \ang{A}_\rho = \tr [ A \rho].\end{equation}
For a chosen orthonormal basis $(\psi_i)_{i} \subseteq \cH$, any operator $A \in \cL(\cH)$ can be represented as a square complex matrix $(\sprod{\psi_i}{A \psi_j})_{i,j}$, i.e.,  in Dirac notation
\begin{equation} A = \sum_{i,j}  \dket{\psi_i} \sprod{\psi_i}{A \psi_j}  \dbra{\psi_j}.\end{equation}
In particular, we have the representation
\begin{equation}  \tr[A] = \sum_{i} \sprod{\psi_i}{A \psi_i},\end{equation}
(which however does not depend on the chosen basis) and the matrix is Hermitian if and only if $A$ is self-adjoint. By the spectral theorem,  one can choose a basis of eigenvectors of $A$, so that
\begin{equation} A = \sum_{i} \lambda_i \dket{\psi_i} \dbra{\psi_i},\end{equation}
where $\lambda_i \in \sigma(A)$ are the eigenvalues of $A$. In particular $\tr[A]$ is the sum (with multiplicities) of the eigenvalues of $A$.

When $A= \rho \in \cS(\cH)$ is a density operator, then $\lambda_i = p_i \in [0,1]$ and one obtains a classical probability density by counting the eigenvalues with their multiplicity, since $\tr[\rho]=1$. A state is called pure if $\rho = \dket{\psi}\dbra{\psi} \in \cS(\cH)$, i.e., its spectrum is only $\cur{0,1}$. One may think of pure states as the quantum counterparts of Dirac delta probabilities at $x_0 \in \cX$, i.e., $x \mapsto 1_{\cur{x = x_0}}$.

\begin{figure}[h]
\begin{center}
\begin{longtable}{|l|l|}
 \caption{Classical notions and their quantum counterparts.\label{table:dictionary}}\\
\hline
Classical & Quantum \\
\hline
$\cX$ (finite set) &  $\cH$ (finite dimensional)\\
$x \in \cX$ &  $\dket{\psi} \in \cH$\\
subset $S \subseteq \cX$ &   subspace $V < H$\\ [0.5ex]
\hline
$f: \cX \to \mathbb{C}$ &  $A: \cH \to \cH$ linear; $A \in \cL(\cH)$\\
$f^*$ &  $A^*: \cH \to \cH$, i.e., the adjoint of $A$\\
$|f|^2$ & $A^*A$\\
$f: \cX \to \R$ &  observable $A=A^*$; $A \in \cO(H)$\\
$f: \cX \to [0, \infty)$ &  $A \in \cO(\cH)$, $\sigma(A) \subseteq [0, \infty)$; $A\ge 0$\\ [0.5ex]
\hline
$\sum_{x \in \cX} f(x)$ &  $\tr[A]$\\
probability $(\rho(x))_{x \in \cX}$ & state $\rho \ge 0$, $\tr[\rho] =1$; $\rho \in \cS(\cH)$\\
Dirac delta $x\mapsto 1_{x_0=x}$ & pure state $\rho = \dket{\psi} \dbra{\psi}$, $\sprod{\psi}{\psi}=1$\\ [0.5ex]

\hline

Cartesian product $\cX \times \cY$ & Tensor product $\cH \otimes \mathcal{K}$\\
Partial sum $\sum_{x} f(x,y)$ & Partial trace $\tr_{\cH}[A]$, $A \in \cL(\cH \otimes \mathcal{K})$\\[0.5ex]
\hline

Shannon  entropy:  & Von Neumann entropy: \\
$S(\rho) = -\sum_{x}\rho(x) \log \rho(x)$ & $S(\rho) = - \tr[ \rho \log \rho ]$\\
Relative entropy: & Quantum relative entropy: \\
$D(\rho||\sigma) = \sum_{x} \rho(x) \ln( \rho(x)/\sigma(x))$ & $S(\rho|| \sigma) = \tr[ \rho (\log \rho - \log \sigma)]$\\ [0.5ex]
\hline
Markov kernel $(N(x,y))_{x \in \cX, y\in \cY}$ & Quantum channel $\Phi: \cL(\cH) \to \cL(\mathcal{K})$\\
\hline
\end{longtable}
\end{center}
\end{figure}

\begin{example}
The simplest example of a non-trivial quantum system is given by a two-dimensional space $\cH = \C^2$, where one denotes the standard basis as $\dket{0} = (1,0)$, $\dket{1} = (0,1) \in \C^2$. This should be seen as the quantum analogue of a two-point space $\cX = \cur{0,1}$ (or of a single bit sequence), hence $\C^2$  is also referred as a single-qubit system. To define the quantum analogue of $n$-bits strings $\cX = \cur{0,1}^n$, we discuss  first the notion of composite quantum systems.
\end{example}

\subsection{Composite systems and partial trace}

A composite quantum system $\cH \otimes \cK$ is defined as the tensor product between two (finite dimensional) systems $\cH$, $\cK$. One often writes $\dket{\psi, \varphi} = \dket{\psi} \otimes \dket{\varphi} \in \cH \otimes \cK$ for the elementary tensor products, whose linear span yields the entire space. The composite system $\cH \otimes \cK$ is endowed with the following scalar product:
\begin{equation} \sprod{ \psi, \varphi}{ \psi', \varphi'}_{\cH \otimes \cK} = \sprod{ \psi}{\psi'}_{\cH} \sprod{\varphi}{\varphi'}_{\cK},\end{equation}
 defined on elementary tensor products and naturally extended by linearity. Given orthonormal bases $(\dket{\psi_i})_{i} \subseteq \cH$, $(\dket{\varphi_j})_{j} \subseteq \cK$, the elements $(\dket{\psi_i, \varphi_j})_{i,j}$ yield an orthonormal basis of $\cH \otimes \cK$, which in particular has dimension $\dim(\cH \otimes \cK) = (\dim \cH) (\dim \cK)$. Writing $A \in \cL( \cH \otimes \cK)$  in matrix form
\begin{equation} A = \sum_{i,j, k, \ell} A_{ij,kl} \dket{\psi_i, \varphi_j} \dbra{\psi_k, \varphi_\ell},\end{equation}
one defines the partial trace of $A$ over $\cH$ as the operator on $\cK$ given by
\begin{equation} \tr_\cH[A] = \sum_{j,\ell} \tr_\cH[A]_{j\ell} \dket{\varphi_j} \dbra{\varphi_\ell},\end{equation}
where
\begin{equation} \tr_\cH[A]_{j\ell} = \sum_{i}  A_{ij,il} .\end{equation}
Similarly,  the partial trace of $A$ over $\cK$ is the operator on $\cH$ given by
\begin{equation} \tr_\cK[A] = \sum_{i,k} \tr_\cH[A]_{ik} \dket{\psi_i} \dbra{\psi_k},\end{equation}
where
\begin{equation} \tr_\cK[A]_{ik} = \sum_{j}  A_{ij,kj}.\end{equation}
It is simple to prove that both $\tr_{\cH}[A]$ and $\tr_{\cH}[A]$ do not depend on the chosen bases. When $A = \rho \in \cS(\cH \otimes \cK)$ is a state, the states obtained as partial traces $\rho_\cK = \tr_\cH[\rho]$, $\rho_\cH = \tr_\cK[\rho]$ play the roles of marginal probabilities and are called reduced density operators.

For composite systems $\bigotimes_{i=1}^n \cH_i$ joining $n$ quantum systems $(\cH_i)_{i=1}^n$, we define similarly the partial trace over the system $\cH_i$, and write
\begin{equation} \label{eq:marginal-i}
\tr_i[A] = \tr_{\cH_i}[A] \in \cL( \bigotimes_{j\neq i} \cH_j).
\end{equation}
 Given $I \subseteq \cur{1, \ldots, n}$, and $\rho \in \cS( \bigotimes_{i=1}^n \cH_i )$ we also write
\begin{equation}\label{eq-marginal-I} \rho_I = \tr_{\cur{1, \ldots, n} \setminus I}[\rho]\end{equation}
for the reduced density operator on the composite sub-system $\bigotimes_{i \in I}\cH_i$.
 In particular, when $\cH_i=  \C^2$, one obtains the system $(\C^2)^{\otimes n}$, the quantum analogue of $n$-bits sequences,  which plays a fundamental role in quantum computing. By taking tensor products of the standard basis elements $\cur{ \dket{0}, \dket{1}}$, one obtains the so-called computational basis
\begin{equation} \cur{ \dket{x} }_{x \in \cur{0,1}^n} \subseteq (\C^2)^{\otimes n}.\end{equation}
Usually, one drastically simplifies the notation writing e.g.\
\begin{equation} \dket{000}  = \dket{0}\otimes \dket{0}\otimes \dket{0}, \quad \dket{001}  = \dket{0}\otimes \dket{0}\otimes \dket{1}, \end{equation}
and similarly for all binary sequences.  Although there are many bases available, the computational basis has a distinctive role. For example, classical probabilities over the set of binary sequences are always seen in correspondence with diagonal states in the computational basis.

\subsection{Quantum channels}
The partial trace operators $\tr_\cH[\cdot]$, $\tr_\cK[\cdot]$ are linear and positive, i.e.~map positive operators (on $\cH \otimes \cK$) into positive operators (respectively, on $\cK$ and $\cH$). In fact, they provide a fundamental example of quantum channels, which are the quantum analogues of classical Markov operators, i.e.\ the operators induced by integration with respect to probability kernels.

A quantum channel from a system $\cH$ into a system $\cK$ can be abstractly defined as linear, completely positive, and trace preserving operator $\Phi: \cL(\cH) \to \cL(\cK)$, where the latter means that
\begin{equation} \tr[ \Phi(A)] = \tr[A] \quad \text{for every $A \in \cL(H)$.}\end{equation}
In particular, $\Phi$ maps states on $\cH$ into states on the sytems $\cK$. We do not enter into the details of the theory, in particular concerning complete positivity, which is a strengthening on the concept of positive operator. We only recall here the fundamental structure result, stating that quantum channels from $\cH$ into $\cK$ are  precisely those linear operators $\Phi: \cL(\cH)  \to \cL(\cK)$ that can be represented via a suitable (finite) family of Kraus operators $(B_i)_{i}$, where each $B_i: \cH \to \cK$ is linear, in the following way:
\begin{equation} \Phi(A) = \sum_{i} B_i A B_i^*\quad \text{for every $A \in \cL(\cH)$,}\end{equation}
and moreover
\begin{equation}\label{eq:sum-kraus} \sum_{i} B_i^* B_i = \I_{\cH}.\end{equation}
Moreover, the number of Kraus operators in the representation is always bounded from above by $(\dim \cH ) (\dim \cK)$.
In particular, both the trace $\tr[\cdot]: \cL(\cH) \to \C = \cL(\C)$ and the partial trace operators enjoy such a Kraus representation. Another example is the conjugation via of a unitary map $U: \cH \to \cK$, i.e., $U^* U = \I_{\cH}$, letting $\Phi(A) = U A U^*$. Using the representation in terms of Kraus operators, one sees immediately that the set of all quantum channels from $\cH$ into $\cK$ define a convex compact subset of all the linear operators  from $\cL(\cH)$ to $\cL(\cK)$. Finally, given a quantum channel $\Phi$ from $\cH$ into $\cK$, its adjoint (with respect to the Hilbert-Schmidt scalar product both on $\cL(\cH)$ and $\cL(\cK)$) is usually denoted by $\Phi^\dagger$ and represented by the adjoint of the Kraus operators associated to $\Phi$:
\begin{equation} \Phi^{\dagger}(A) = \sum_{i} B_i^* A B_i,\quad \text{for every $A \in \cL(\cK)$.}\end{equation}
Notice that, by \eqref{eq:sum-kraus}, $\Phi^\dagger(\I_{\cK}) = \I_{\cH}$, i.e., $\Phi^\dagger$ is unital.

Not all the linear operators that map states into states are quantum channels, i.e., enjoy a representation in terms of Kraus operators. A natural example is provided by the transpose map. Given $A \in \cL(\cH)$, its transpose $A^T \in \cL(\cH^*)$ is defined as
\begin{equation} \cH ^* \ni \dbra{\psi} \mapsto A^T \dbra{\psi} = \dbra{\psi} A,\end{equation}
where $\dbra{\psi}A (\dket{\varphi}) = \sprod{\psi}{A \varphi}$ for every $\dket{\varphi} \in \cH$. Given any orthonormal basis $(\dket{\psi_i})_i \subseteq \cH$, the matrix representing $A^T$ with respect to the dual basis $(\dbra{\psi_i})_i \subseteq \cH ^*$ is indeed the transpose (but not conjugate) of the matrix representing $A$. Since the spectrum of $A$ equals the spectrum of its transpose, we see immediately that the transpose maps $\cS(\cH)$ into $\cS(\cH^*)$. However, as soon as $\dim(\cH) \ge 2$ it is possible to prove that it is not a quantum channel.

\begin{exercise}
On a single-qubit quantum system $\cH = \C^2$, for $p \in [0,1]$, define
\begin{equation} \rho_p = (1-p) \dket{0}\dbra{0} + p \dket{1}\dbra{1},\end{equation}
and
\begin{equation}  \dket{\psi_p} = \sqrt{ (1-p)} \dket{0}+ \sqrt{p} \dket{1}.\end{equation}
Prove that $\rho_p$, $\dket{\psi_p} \dbra{\psi_p} \in \cS(\C^2)$. Are they equal?
\end{exercise}

\begin{exercise}
On a two-qubit quantum system $\C^{2} \otimes \C^2$, consider
\begin{equation}  \dket{\Phi^+} := \frac{1}{\sqrt{2}} \bra{ \dket{00} + \dket{11} }\end{equation}
and the associated Bell state $\pi = \dket{\Phi^+} \dbra{\Phi^+} \in \cS(\C^2 \otimes \C^2)$ (which is a pure state). Compute the reduced density operators on the two single-qubit subsystems.
\end{exercise}

\section{An overview of Quantum Optimal Transport}\label{sec:overview}

The development of quantum computing and quantum information theory has led to the exploration of various mathematical concepts, including quantum analogues of classical distances on probabilities. Given states $\sigma, \rho \in \cS(\cH)$ on a (finite dimensional) quantum system $\cH$,

\begin{enumerate}
\item  the quantum analogue of the total variation distance \eqref{eq:tv} is called the trace distance, and defined  as
\begin{equation} D_{\tr} (\sigma, \rho) = \frac 12 \tr[ |\rho- \sigma|],\end{equation}
where $|\rho -\sigma|$ is defined via functional calculus, so that $\tr[ | \rho -\sigma|]$ is the sum (with multiplicities) of the moduli of the eigenvalues of $\rho -\sigma \in \cO(\cH)$,
\item the quantum analogue of the Hellinger distance \eqref{eq:hellinger} is related to the quantum fidelity, defined as
\begin{equation} F(\sigma, \rho) = \tr\sqa{ \sqrt{  \sqrt{\rho} \sigma \sqrt{\rho} } }^2,\end{equation}
while the actual distance corresponding to the Hellinger distance is the Bures metric $\sqrt{(1- F(\sigma, \rho))/2}$,
\item the quantum analogue of the relative entropy \eqref{eq:kl} is the (Umegaki) quantum relative entropy, defined as
\begin{equation} S(\rho \|\sigma) = \tr[ \rho ( \log \rho - \log \sigma) ],\end{equation}
whenever the kernel of $\sigma$ is contained in the kernel of $\rho$ (otherwise one sets the relative entropy to be $+\infty$, as in the classical case).
\end{enumerate}
Like their classical counterparts, these distances can be defined on quite general systems and can be  computed or approximated with a relatively small effort, taking into account the dimension of the system. Moreover, they are not adapted to any specific geometry in the underlying space. This can be made precise by noticing that they are unitarily invariant, i.e., for every unitary $U: \cH \to \cH$,
\begin{equation}  D_{\tr} (U\sigma U^*, U\rho U^*) =   D_{\tr} (\sigma, \rho),\end{equation}
and similarly for the quantum fidelity and the relative entropy. More generally, all these quantities are monotone with respect to the action of any quantum channel $\Phi$ from $\cH$ into $\cK$, i.e.,
\begin{equation}\label{eq:trace-data-processing} D_{\tr} (\Phi(\sigma), \Phi(\rho) ) \le D_{\tr} (\sigma, \rho),\end{equation}
and similarly for the relative entropy
\begin{equation}\label{eq:entropy-data-processing}
S(\Phi(\rho) \| \Phi(\sigma) ) \le  S(\rho \| \sigma)
\end{equation}
and the quantum fidelity (with a reverse inequality). Such inequalities are the quantum analogues of their classical counterparts.

Several proposals for quantum optimal transport problems and induced distances between quantum states have been put forward in recent years, with various applications. Although these lecture notes are focused on two recent proposals by the authors and their collaborators \cite{de2021quantum-channel, de2021quantum-1}, in the literature, many alternative approaches are available, and the (possible) links between them are still not completely explored. The earliest  formulation dates back to Connes and Lott in 1992 \cite{connes1992metric}, who defined the spectral distance in non-commutative geometry. Another early approach was put forth by Zyczkowski and Slomczynski in 1997 \cite{karol1998monge}, who computed the Wasserstein distance between the Husimi (classical) probability distributions associated to states in Bosonic systems. In the setting of free probability, Biane and Voiculescu proposed an analogue of the Wasserstein matric in \cite{biane2001free} . Since 2012, Maas and Carlen \cite{carlen2014analog, carlen2017gradient, carlen2020non} have devised a distance formulating a quantum analogue of the classical Benamou-Brenier formula, which gives a continuous-time formulation of the optimal transport problem (for more information, see also the lecture notes by E.\ Carlen in this volume). In 2013, Agredo \cite{agredo2013wasserstein} proposed a Wasserstein distance that extends any given distance on a set of basis vectors. In 2016, Golse, Mouhot and Paul introduced a quantum Kantorovich problem using quantum couplings, with applications to semi-classical limits \cite{golse2016mean, caglioti2020quantum, golse2021quantum, caglioti2021towards} (for more information, see also the lecture notes by F.\ Golse in this volume).

We can roughly classify all these definitions according to the point of view on classical optimal transport they mostly emphasize, i.e., the (primal) Monge-Kantorovich problem, the dual problem, or the Benamou-Brenier formulation, see Table \ref{table:review}.

\noindent\begin{minipage}{\textwidth}
\begin{center}
\begin{longtable}{|p{0.8\textwidth}|}
 \caption{Classifying quantum optimal transport theories.\label{table:review}}\\
\hline
Monge-Kantorovich \\
\hline
\begin{itemize}[label={-},nosep]
\item distance between Husimi functions \cite{karol1998monge}
\item transport via couplings \cite{golse2016mean}
\item transport via channels \cite{de2021quantum-channel}
\end{itemize}
\\ [0.5ex]
\hline
 Dual problem \\
 \hline

   \begin{itemize}[label={-},nosep]
 \item  spectral distance \cite{connes1992metric}
 \item distance on a basis \cite{agredo2013wasserstein}
 \item Wasserstein distance of order $1$ \cite{de2021quantum-1}
 \end{itemize}
 \\ [0.5ex]
 \hline
 Benamou-Brenier \\
 \hline

  \begin{itemize}[label={-},nosep]
 \item  quantum Benamou-Brenier \cite{carlen2014analog}
 \end{itemize}
 \\
 \hline
\end{longtable}
\end{center}

\end{minipage}

\section{Quantum optimal transport via quantum channels}\label{sec:channels}

In this section, we focus on the proposal \cite{de2021quantum-channel}, where a notion of optimal transport between quantum states is formulated using quantum channels.

\subsection{Quantum transport plans} Following the analogy with the classical theory, we  write $\cX$, $\cY$ for two finite-dimensional quantum systems (possibly $\cX = \cY$) and let $\sigma \in \cS(\cX)$, $\rho \in \cS(\cY)$ be density operators. Recalling the notion of coupling in the sense of Kantorovich, it seems quite natural to define a quantum coupling as a state $\pi \in \cS( \cX \otimes \cY)$  such that
\begin{equation}\label{eq:coupling-golse} \tr_{\cY}[\pi] = \sigma, \quad \tr_{\cX}[\pi] = \rho.\end{equation}
If a cost function is replaced by an operator $C \in \cO(\cX \otimes \cY)$ and the average cost of a
 coupling is $\ang{C}_\pi = \tr[C \pi]$, this strategy yields a straightforward notion of quantum optimal transport cost, by optimizing the  average cost among all the couplings.
Indeed, this is the point of view developed in \cite{golse2016mean} with applications e.g.\ to semi-classical limits of interacting systems.

However, a different approach can also be devised. Indeed, in the classical theory, i.e., when $\cX$ and $\cY$ are sets, Kantorovich couplings $\pi(x,y)$ can be interpreted as generalized maps (transport plans) by conditioning, e.g.\ with respect to the first variable, and defining
\begin{equation} \pi(y| x) = \frac{ \pi(x,y)}{\sigma(x)},\end{equation}
provided that $\sigma(x) >0$. This defines a Markov kernel pushing $\sigma$ into $\rho$.

The quantum analogues of Markov kernels are quantum channels $\Phi$, so we begin with the following definition.

\begin{definition}
Given $\sigma \in \cS(\cX)$, $\rho \in \cS(\cY)$ a quantum transport plan $\Phi$ from $\sigma$ to $\rho$ is a quantum channel $\Phi$ from $\cX$ to $\cY$ such that $\Phi(\sigma) = \rho$.
\end{definition}
 However, we also need to define a notion of transport cost. The issue is that we apparently need a suitable state in a composite system where both states $\sigma$ and $\Phi(\sigma) = \rho$ are encoded. The strategy is to rely first upon the so-called purification of a quantum state to build two ``copies'' of $\sigma$ and then act with $\Phi$ only on one such copy.

Let us prove the following well-known result (valid in fact also for infinite dimensional systems).

\begin{proposition}[purification of a state]\label{prop:purification}
Given any $\sigma \in \cS(\cH)$, there exists an auxiliary quantum system $\cK$ and a pure state $\dket{\Psi}\dbra{\Psi} \in \cS(\cH \otimes \cK)$ such that
\begin{equation} \tr_{\cK}[ \dket{\Psi}\dbra{\Psi}] = \sigma.\end{equation}
\end{proposition}

\begin{proof}
Let $\cK = \cH^*$ be the dual of $\cH$, and consider the canonical isomorphism
\begin{equation} \dket{\psi} \otimes \dbra{\varphi} \mapsto  \dket{\psi} \dbra{\varphi}\end{equation}
between $\cH \otimes \cH^*$ and $\cL(\cH)$. We can thus find $\dket{\Psi} \in \cH \otimes \cK$ corresponding to the operator $\sqrt{\sigma} \in \cL(\cH)$ (defined via spectral calculus). We claim that such $\dket{\Psi}$ satisfies the thesis. Indeed, by choosing an orthonormal basis $(\dket{\psi_i})_{i }\subseteq \cH$ of eigenvectors of $\sigma$ with corresponding eigenvalues $(p_i)_{i}$, we have
\begin{equation} \sqrt{ \sigma} = \sum_{i} \sqrt{p_i} \dket{\psi_i} \dbra{\psi_i},\end{equation}
hence
\begin{equation} \dket{\Psi} = \sum_{i} \sqrt{p_i}  \dket{\psi_i}\otimes \dbra{\psi_i}.\end{equation}
Thus,
\begin{equation} \dket{\Psi} \dbra{\Psi} = \sum_{i,j} \sqrt{p_i p_j}  (\dket{\psi_i}\otimes \dbra{\psi_i})( \dbra{\psi_j} \otimes \dket{\psi_j}),\end{equation}
and taking the partial trace over $\cK$ yields
\begin{equation} \tr_\cK[\dket{\Psi} \dbra{\Psi} ] = \sum_{ i}p_i \dket{\psi_i} \dbra{\psi_i} = \sigma.\qedhere \end{equation}
\end{proof}

Of course, neither $\dket{\Psi}$ nor $\cK$ is by any means unique. We thus refer to the construction in the proof above as the canonical purification of $\rho$, so that $\cK = \cH^*$ and one can check easily that
\begin{equation} \tr_{\cH}[ \dket{\Psi} \dbra{\Psi}] = \sum_{ i}p_i \dbra{\psi_i} \,  \dket{\psi_i} \in \cL(H^*) = \sigma^T\end{equation}
is the transpose operator of $\sigma$. %An example that will useful later is to let $\cK = \cH^* \otimes \cH \otimes \cH^*$ and define...

\begin{example}
On a single-qubit system $\cH = \mathbb{C}^2$, the state
\begin{equation} \sigma =\frac 1 2 \I_{\C^{2}} =  \frac 1 2|0\rangle \langle 0| +\frac 1 2 |1 \rangle \langle 1 |\end{equation}
admits as a purification the Bell state
\begin{equation}|\Phi^+\rangle = \frac{1}{\sqrt{2}}|00\rangle  + \frac{1}{\sqrt{2}}|11\rangle \in (\mathbb{C}^{2})^{\otimes 2 }.\end{equation}
Notice that the above is the canonical purification, up to identifying $\mathbb{C}^2$ with its dual.
\end{example}

Let us also notice that an analogous concept to that of purification cannot exist in  classical probability, since pure states correspond to Dirac deltas, hence their marginals will be Dirac deltas too. Indeed, the purifications will yield (in general) so-called entangled states, which is a peculiar quantum property that has no classical counterpart.

Back to the optimal transport problem, given the states $\sigma \in \cS(\cX)$ and $\rho \in \cS(\cY)$, we associate to any quantum transport plan $\Phi$ from $\sigma$ to $\rho$, the density operator
\begin{equation} \pi_\Phi = ( \Phi\otimes \I_{\cL(\cX^*)} ) (\dket{\Psi}\dbra{\Psi}) \in \cS(\cY \otimes \cX^*),\end{equation}
where $\dket{\Psi} \in \cX \otimes \cX^*$ denotes the canonical purification of $\sigma$. The linear operator $\Phi\otimes \I_{\cL(\cX^*)}$ acts on elementary tensor products $A \otimes B \in \cL( \cX \otimes \cX^*)$ as follows:
\begin{equation} \Phi  \otimes \I_{\cL(\cX^*)} ( A \otimes B) = \Phi(A) \otimes B\end{equation}
and is then extended by linearity to the whole $\cL( \cX \otimes \cX^*) = \cL(\cX) \otimes \cL(\cX^*)$. Using that $\Phi$ admits a representation in terms of Kraus operators, one sees immediately that also $\Phi\otimes \I_{\cL(\cX^*)}$ admits a similar representation, hence it is also a quantum channel and in particular when applied to $\dket{\Phi}\dbra{\Phi}$, it yields indeed a quantum state. Since the channel acts only on the system $\cX$, it is not difficult to prove that reduced density operators of $\pi_\Phi$ are
\begin{equation} \tr_\cY[ \pi_\Phi] = \sigma^T, \quad \tr_{\cX^*} [ \pi_\Phi]  = \Phi(\sigma) = \rho.\end{equation}

From this correspondence, we are thus led to a notion of quantum coupling different than  \eqref{eq:coupling-golse}.

\begin{definition}
Given $\sigma \in \cS(\cX)$, $\rho \in \cS(\cY)$, a quantum coupling between $\sigma$ and $\rho$ is any state $\pi \in \mathcal{S}(\cY \otimes \cX^*)$ such that
\begin{equation} \tr_{\cY} [ \pi ] = \sigma^T \quad \quad \tr_{\cX^*} [ \pi ] = \rho.\end{equation}
The set of such quantum couplings is denoted with $\cop(\sigma, \rho)$.
\end{definition}
Some examples:
\begin{itemize}[label=-]
\item
As in the classical case, the set of quantum couplings is always non-empty, since it contains the product coupling $\pi = \rho \otimes \sigma^T$, induced by the trivial quantum channel
\begin{equation} \Phi(A) =  \tr[A]\rho,\end{equation}
which maps any state (in particular $\sigma$) into $\rho$.
\item If $\sigma  = \dket{\psi}\dbra{\psi}$ is pure, then the product coupling is the only coupling, hence $\cop(\sigma, \rho) = \cur{\rho \otimes \sigma^T}$.
\item If $\rho = \sigma$ (hence $\cX = \cY$) the coupling  induced by the identity  channel $\Phi(A) = A$, which corresponds to the canonical purification of $\sigma$, belongs to $\cop(\sigma, \sigma)$.
\end{itemize}

\begin{remark}\label{rem:symmetry}
Although the definition of  transport plans is not symmetric, since we require that $\Phi(\sigma) = \rho$, symmetry is restored at the level of  couplings, since $\cop(\sigma, \rho)$ is in natural correspondence with $\cop(\rho, \sigma)$, via a swap-transpose map, acting on elementary tensor products as
%\begin{equation} |\phi\rangle \otimes \langle \psi | \mapsto |\psi\rangle \otimes \langle \phi|,\end{equation}
\begin{equation}
\alpha\otimes\beta^T \mapsto \beta \otimes \alpha^T\qquad \forall\;\alpha\in\mathcal{L}(\mathcal{Y})\,,\;\beta\in\mathcal{L}(\mathcal{X})\,,
\end{equation}
and extended by linearity to $\mathcal{L}(\cY) \otimes \mathcal{L}(\cX^*)$.
\end{remark}

The set of quantum couplings is in correspondence with quantum plans, as the following construction shows (see \cite{de2021quantum-channel} for details). Given $\Pi \in \cop(\sigma, \rho) \subseteq \mathcal{S}(\cY \otimes \cX^*)$, by the spectral theorem, write
\begin{equation}\label{eq:Pi-coupling-channel} \Pi = \sum_i p_i \dket{\Psi_i}\dbra{\Psi_i}\end{equation}
with $(\dket{\Psi_i})_i\subseteq \cY \otimes \cX^*$ is an orthonormal basis and $\sum_i p_i = 1$. Using the canonical isomorphism between $\cY \otimes \cX^*$ and the space of linear operators from $\cX$ to $\cY$,
\begin{equation} \dket{\psi} \otimes \dbra{\varphi} \mapsto  \dket{\psi} \dbra{\varphi},\end{equation}
we have an orthonormal basis $(F_i)_{i}$ corresponding to $(\dket{\Psi_i})_i$. The identity \eqref{eq:Pi-coupling-channel} yields that
\begin{equation}\label{eq:F_i-properties} \sum_{i} p_i F_i^* F_i = \left(\mathrm{tr}_{\mathcal{Y}}\Pi\right)^T = \sigma, \quad \text{and} \quad \sum_{i} p_i F_i F_i^* = \mathrm{tr}_{\mathcal{X}^*}\Pi = \rho.\end{equation}
We thus define the linear operator $\Phi: \cL(\cX) \to \cL(\cY)$,
\begin{equation} \Phi (A) = \sum_i \sqrt{p_i}\,F_i\,\sigma^{-\frac{1}{2}} \, A \, \sigma^{-\frac{1}{2}}\,F_i^* \sqrt{p_i}, \end{equation}
where for simplicity we have assumed that $\sigma$ is invertible (otherwise one needs to use a pseudo-inverse). %Since
Defining the Kraus operators $B_i = \sqrt{p_i}\,F_i\,\sigma^{-\frac{1}{2}}$, using \eqref{eq:F_i-properties} one obtains
\begin{equation} \sum_{i} B_i^*B_i = \sum_{i} \sigma^{-\frac{1}{2}}\,F_i^* \sqrt{p_i} \sqrt{p_i}\,F_i\,\sigma^{-\frac{1}{2}} = \I_{\cX},\end{equation}
and
\begin{equation} \Phi(\sigma) = \sum_i \sqrt{p_i}\,F_i\,\sigma^{-\frac{1}{2}} \, \sigma \, \sigma^{-\frac{1}{2}}\,F_i^* \sqrt{p_i} = \sum_{i} p_i F_i F_i^*  = \rho,\end{equation}
hence $\Phi$ is a quantum channel such that $\Phi(\sigma) = \rho$.

\subsection{Quantum optimal transport cost}

We are now in a position to follow the same path as in \cite{golse2016mean}, with this alternative notion of coupling. Given states $\sigma \in \cS(\cX)$, $\rho \in \cS(\cY)$ and a cost operator $C \in \cO(\cY \otimes \cX^*)$, we search for the coupling $\pi \in \cop(\sigma, \rho)$ which minimizes the average cost $\ang{C}_\pi = \tr[C \pi]$.
%\begin{equation} \min_{\pi \in \cop(\sigma, \rho) } \tr[ C \pi].\end{equation}
Since the set of couplings is closed and convex and the average cost is linear, an optimal coupling $\pi$ always  exists (operationally, it is a semidefinite programming problem).

In the remainder of this section, we focus on the particular case of \emph{quadratic} cost, i.e.\ we model the cost after the squared Euclidean distance in $\R^d$. Precisely, we let $\cX = \cY$, we fix a set of $d$ observables $(R_i)_{i=1}^d \subseteq \cO(\cX)$  and define the cost operator
\begin{equation} C  = \sum_{i=1}^d ( R_i \otimes \I_{\cX^*} - \I_\cX \otimes R_i^T)^2.\end{equation}
We then write
\begin{equation} D(\sigma, \rho)^2 = \min_{\pi \in \cop(\sigma, \rho) } \tr[ C \pi],\end{equation}
which plays the role of a squared Wasserstein distance of order $2$ in this quantum setting, where $(R_i)_{i=1}^d$ are the ``directions'' that are used to measure the transport cost. Although this is not necessarily a distance, we are going to see that it enjoys several natural properties.

First, let us notice that the structure of the cost operator allows for rewriting  $\ang{C}_\pi$ for any coupling $\pi$ in terms of the associated transport plan $\Phi$, i.e.\ such that
\begin{equation} \pi = \left( \Phi \otimes \I_{\cL(\cX^*)}\right)  \bra{\dket{\Psi} \dbra{\Psi}},\end{equation}
where $\dket{\Psi} \in \cX \otimes \cX^*$ is the canonical purification of $\sigma$. Writing $\Phi^\dagger$ for the adjoint of the channel $\Phi$, we have
\begin{equation}  \tr[ \pi C ]  = \tr\sqa{ \dket{\Psi}\dbra{\Psi} \bra{ \Phi^\dagger \otimes \I_{\cL(\cX^*)} } (C) }
 = \left\langle\Psi\left|\left( \Phi^\dagger \otimes \I_{\cL(\cX^*)} \right) (C) \right|  \Psi\right\rangle .
\end{equation}
By expanding  each square in the the definition of $C$, we have
\begin{equation} (R_i\otimes \I_{\cX^*} - \I_\cX\otimes R_i^T)^2  =  R_i^2\otimes \I_{\cX^*} + \I_\cX\otimes(R_i^2)^T - 2 R_i \otimes R_i^T,\end{equation}
so that
\begin{align}
&\bra{\Phi^\dagger \otimes \I_{\cL(\cX^*)}}\left( (R_i\otimes \I_{\cX^*}  - \I_\cX\otimes R_i^T)^2 \right) \nonumber\\
& =  \Phi^\dagger(R_i^2)\otimes \I_{\cX^*}   + \I_\cX\otimes (R_i^2)^T - 2 \Phi^{\dagger}(R_i) \otimes R_i^T.\end{align}
Using that $\Psi$ is the canonical purification of $\sigma$, we have that
\begin{equation} \left\langle\Psi\left| \Phi^\dagger(R_i^2)\otimes \I_{\cX^*} \right| \Psi\right\rangle = \tr[  \Phi^\dagger(R_i^2) \sigma ] = \tr[ R_i^2 \Phi(\sigma)] = \tr[R_i^2 \rho].\end{equation}
Similarly,
\begin{equation}  \left\langle\Psi\left|\I_\cX\otimes (R_i^2)^T \right| \Psi \right\rangle = \tr[  (R_i^2)^T \sigma^T ] = \tr[R_i^2 \sigma].\end{equation}
These terms do not depend on the specific transport plan (or the coupling). Of course, the third term instead does depend on $\Phi$. Recalling that $\dket{\Psi} \in \cX \otimes \cX^*$ corresponds to $\sqrt{\sigma} \in \cL(\cX)$ in the construction of canonical purification Proposition~\ref{prop:purification}, we have the identity
\begin{equation}  \sprod{\Psi} { \Phi^{\dagger}(R_i) \otimes R_i^T | \Psi} =  \tr[\sqrt{\sigma} \,\Phi^{\dagger}(R_i) \, \sqrt{\sigma} \, R_i ].\end{equation}
Summing all these contributions, we have  the equivalent expression that uses only the transport plan $\Phi$:
\begin{equation}\label{eq:C-pi-plan} \ang{C}_\pi  = \sum_{i=1}^d \left(\tr[  R_i^2 \sigma ] + \tr[  R_i^2 \rho] -2  \tr[ R_i\, \sqrt{\sigma} \, \Phi^{\dagger}(R_i)\,  \sqrt{\sigma}]\right),\end{equation}
which shows that minimizing $\ang{C}_\pi$ is equivalent to maximizing the correlation-like quantity
\begin{equation} \sum_{i=1}^d  \tr[ R_i\, \sqrt{\sigma}\, \Phi^{\dagger}(R_i) \, \sqrt{\sigma}].\end{equation}

\subsection{Properties}
Let us obtain some properties of $D(\sigma, \rho)$. First, notice that by Remark~\ref{rem:symmetry} and the structure of the cost $C$, it follows easily that $D(\sigma, \rho) = D(\rho, \sigma)$.

Notice that, if $\sigma = \rho$ and $\Phi = \I_{\cL(\cX)}$ is the identity channel, then the average cost becomes
\begin{equation} 2 \sum_{i=1}^d\left(  \tr[  R_i^2\sigma ] -  \tr[ R_i \sqrt{\sigma} R_i  \sqrt{\sigma}]\right). \end{equation}
We can prove that this quantity is indeed $D^2(\sigma, \sigma)$, i.e.\ the identity channel is always an optimal plan, as one would expect. Indeed, for general $\sigma$, $\rho \in \cS(\cX)$, we establish the inequality
\begin{equation}\label{eq:distance-lower-bound} D^2(\sigma, \rho) \ge \sum_{i=1}^d \left( \tr[  R_i^2 \sigma] - \tr[ R_i \sqrt{\sigma} R_i  \sqrt{\sigma}] +\tr[R_i^2 \rho] -  \tr[ R_i \sqrt{\rho} R_i  \sqrt{\rho}]\right),\end{equation}
which for $\sigma = \rho$ becomes an equality (since in the minimization that defines the left-hand side one can take the identity channel). To prove \eqref{eq:distance-lower-bound}, by \eqref{eq:C-pi-plan} it is sufficient to consider any transport plan $\Phi$ from $\sigma$ to $\rho$ and argue that, for every $i$,
\begin{equation} 2 \tr[ R_i \sqrt{\sigma} \Phi^{\dagger}(R_i)  \sqrt{\sigma}] \le  \tr[ R_i \sqrt{\sigma} R_i  \sqrt{\sigma}] + \tr[ R_i \sqrt{\rho} R_i  \sqrt{\rho}].\end{equation}
Using the inequality
\begin{equation} 2 \tr[AB] \le \tr[A^2] +\tr[B^2], \quad \text{for $A$, $B \in \cO(\cX)$,}\end{equation}
with $A = \sigma^{1/4} R_i \sigma^{1/4}$,  $B = \sigma^{1/4}\Phi^{\dagger}(R_i)\sigma^{1/4}$, the thesis reduces to the validity of
\begin{equation} \tr[ \Phi^{\dagger}(R_i) \sqrt{\sigma} \Phi^{\dagger}(R_i)  \sqrt{\sigma}] \le \tr[ R_i \sqrt{\rho} R_i \sqrt{\rho}]. \end{equation}
Since $\rho = \Phi(\sigma)$, this is a direct application of the monotonicity version of Lieb's concavity theorem (see \cite[Theorem 1.6]{carlen2022some} with $t=1/2$, but also E.\ Carlen's lecture notes in this volume).

Inequality \eqref{eq:distance-lower-bound} with $\sigma = \rho$ yields that $D(\sigma, \sigma)^2$ can be in fact strictly positive -- simple examples can be found, e.g.\ in the single-qubit case \cite{geher2023quantum}.
Furthermore, \eqref{eq:distance-lower-bound} implies that for any two quantum states $\rho,\,\sigma\in\mathcal{S}(\mathcal{H})$ we have
\begin{equation}\label{eq:Dav}
D(\rho, \sigma)^2 \ge \frac 1 2 D(\rho, \rho)^2 + \frac 1 2 D(\sigma, \sigma)^2\,.
\end{equation}

 The second interesting property that we mention, is a modified triangle inequality:
\begin{equation}\label{eq:modified-triangle} D(\sigma, \rho) \le D(\sigma, \tau) + D(\tau, \tau) + D(\tau, \rho) \qquad \forall\;\rho\,,\;\sigma\,,\;\tau \in \mathcal{S}(\cX).\end{equation}
The proof of the triangle inequality for the classical Wasserstein distance uses a \emph{gluing} procedure between couplings, which however uses conditioning. In terms of transport plans, however, it is clear that starting from any $\Phi_{\sigma \to \tau}$ (from $\sigma$ to $\tau$) and $\Phi_{\tau \to \rho}$ (from $\tau$ to $\rho$), the composition $\Phi_{\tau \to \rho} \circ \Phi_{\sigma \to \tau}$ yields a plan from $\sigma$ to $\rho$, which should provide an upper bound to $D(\sigma, \rho)$. The actual proof however is more involved, and we refer to \cite[Theorem 2]{de2021quantum-channel} for details, explaining in particular the appearance of the extra term $D(\tau, \tau)$ which again may be strictly positive.

\begin{remark}
In view of \eqref{eq:Dav}, we conjecture that the quantity
\begin{equation}  \sqrt{  D(\rho, \sigma)^2 - \frac 1 2 D(\rho, \rho)^2 - \frac 1 2 D(\sigma, \sigma)^2}\end{equation}
may indeed be an actual distance (up to some non-degeneracy assumptions on the $R_i$'s to ensure that the quantity is null if and only if $\rho = \sigma$).
\end{remark}

\section{The quantum Wasserstein distance of order 1 for qubits}\label{sec:qubits}

In this section, following \cite{de2021quantum-1}, we discuss a notion of quantum Wasserstein distance of order $1$, generalizing the classical one induced by the Hamming distance. To keep exposition simple, we limit ourselves to $n$-qubit systems, i.e.\ $(\C^2)^{\otimes n}$, which play a fundamental role in quantum computing and quantum information theory.

\subsection{The Hamming-Wasserstein distance}
Before we describe the quantum Wasserstein distance of order $1$, let us briefly review the classical case.

The set $\cX = \cur{0,1}^n$ of binary strings of length $n$ (also called the discrete $n$-cube) can be endowed with the Hamming distance
\begin{equation} |x-y|_H = \sum_{i=1}^n |x_i-y_i| = \sum_{i=1}^n 1_{\{x_i \neq y_i\}},\end{equation}
 for $x = (x_i)_{i=1}^n$, $y = (y_i)_{i=1}^n \in \cX$. Such a distance induces by \eqref{eq:w1} the Wasserstein distance of order $1$ between probability distributions $\sigma$, $\rho$  over strings
\begin{equation} W_1(\sigma, \rho) = \min_{\pi \in \cop(\sigma, \rho)} \sum_{x,y \in \{0,1\}^n} |x-y|_H \pi(x,y),\end{equation}
which roughly speaking measures the optimal number of characters in the string that on average one has to change to transform $\sigma$ into $\rho$. For example, if  $\sigma$ is a uniform distribution over the possible $2^n$ binary strings, and $\rho$ is a Dirac delta at the string consisting of all zeros, then the only possible coupling leads to
\begin{equation} W_1(\sigma, \rho) = \frac 1 {2^n} \sum_{x \in \{0,1\}^n} \sum_{i=1}^n |x_i| = \frac{n}{2}.\end{equation}

The dual formulation \eqref{eq:kantorovich-classical} reads in this case
\begin{equation} W_1(\sigma, \rho) = \max_{f} \left\{ \sum_{x \in \{ 0,1\}^n} f(x) (\rho(x) - \sigma(x))\right\},\end{equation}
where maximization runs among all functions $f :\{ 0,1\}^n \to \R$ that are $1$-Lipschitz with respect to the Hamming distance, i.e.,
\begin{equation} |f(x) - f(y)|  \le |x-y|_H  = \sum_{i=1}^n 1_{\{x_i \neq y_i\}} \quad \text{for every $x, y \in \cur{0,1}^n$.} \end{equation}
Arguing either from the primal or the dual formulation, it is not difficult to  conclude that, if $\sigma$ and $\rho$ have the same $(n-1)$-marginal distribution, e.g.\
\begin{equation} \sigma(0, x) + \sigma(1, x) = \rho(0, x) + \rho(1, x)  \quad \text{for every $x \in \{0,1\}^{n-1}$,}\end{equation}
then
\begin{equation} W_1(\sigma, \rho) \le 1.\end{equation}
This is because one can build a coupling $\pi \in \cop(\sigma, \rho)$ which keeps unchanged all letters but one. More generally, if for some $1 \le k \le n$, $\sigma$ and $\rho$ have the same $(n-k)$-marginal distributions, then
\begin{equation} W_1(\sigma, \rho) \le k.\end{equation}

Finally, the following inequalities hold, for probability distributions $\sigma$, $\rho$ on $\{0,1\}^n$:
\begin{equation}\label{eq:comparison-TV-W1} \| \sigma - \rho\|_{TV} \le W_1(\sigma, \rho) \le n \| \sigma - \rho \|_{TV}.\end{equation}
This can be seen by combining Exercise~\ref{ex:TV-W} with the trivial inequality
\begin{equation} 1_{\cur{x \neq y}} \le |x-y|_{H} \le n 1_{\cur{x \neq y}}.\end{equation}

\subsection{Construction of the distance}
Back to the quantum case, i.e., $\cX = (\C^{2})^{\otimes n}$, our aim is to define a suitable notion of quantum Wasserstein distance of order $1$, recovering the classical distance for classical states (i.e., density operators that are diagonal with respect to the computational basis) and enjoying useful properties, such as \eqref{eq:comparison-TV-W1}. This eventually leads, by duality, to a notion of quantum Lipschitz observables with interesting features. % (e.g.\ their spectrum concentrates around the expectation \ref{thm:concentration}).

The definition of the quantum Wasserstein distance of order $1$ between two states $\sigma$, $\rho \in \cS((\C^{2})^{\otimes n})$ relies upon the fact that, also in the classical case, it is induced by a norm, and indeed we write $\nor{\rho - \sigma}_{W_1}$ in the quantum case to emphasize this fact. We proceed in multiple steps.
\begin{enumerate}
\item First, mimicking the classical case, we postulate that a state $\rho$ is at distance $\le 1$ from a state $\sigma$ if there exists $i =1,\ldots, n$ such that (recalling the notation \eqref{eq:marginal-i})
\begin{equation} \tr_i[\sigma] = \tr_i[\rho],\end{equation}
i.e.\ the reduced density operators coincide after discarding the $i$-th qubit of the system.
\item Then, we define the unit ball $\mathcal{B}_n$ (centred at $0$) as the convex envelope of differences of states at distance less than one, i.e.,
\begin{equation} \begin{split} \mathcal{B}_n =  & \bigg\{ \sum_{i=1}^n p_i\left(\rho^{(i)} - \sigma^{(i)}\right) :  p_i\ge0,\;\sum_{i=1}^np_i=1, \\
 & \rho^{(i)},\,\sigma^{(i)} \in \cS\bra{(\mathbb{C}^2)^{\otimes n}} \text{ are such that }  \tr_i[\rho^{(i)}]=\tr_i[\sigma^{(i)}] \bigg\}
 \end{split} \end{equation}
and induce a norm via the Minkowski functional:
\begin{equation}
\left\|X\right\|_{W_1} = \min\left(t\ge0:X\in t\,\mathcal{B}_n\right).
\end{equation}
\item Finally, we define the distance between states $\sigma$, $\rho$, as $\nor{\rho - \sigma}_{W_1}$.
\end{enumerate}

Following this construction we are led to the following quantity, which operationally is a semi-definite programming problem.

\begin{definition}[quantum Wasserstein distance of order $1$ on $n$ qubits]
Let $\sigma$, $\rho \in \cS( (\C^2)^{\otimes n} )$. Then, the quantum Wasserstein distance of order $1$ between $\sigma$ and $\rho$ is defined as the quantity
\begin{equation}\label{eq:quantum-w1}\begin{split}
   \|\sigma-\rho\|_{W_1} = & \min \bigg\{ \sum_{i=1}^n c_i : \sigma - \rho = \sum_{i=1}^n c_i( \sigma^{(i)} -\rho^{(i)}), c_i \ge 0, \\
   & \quad  \sigma^{(i)}, \rho^{(i)} \in \cS\bra{ (\C^2)^{\otimes n}} \text{ are such that }   \tr_i[\sigma^{(i)}] = \tr_i[\rho^{(i)}] \bigg\}.\end{split}\end{equation}
\end{definition}
where the minimum runs among all possible $c_i$, $\sigma^{(i)}$ and $\rho^{(i)}$'s.

\subsection{Basic properties}

From the very definition, we see that the quantum Wasserstein distance of order $1$ is invariant with respect to permutations of the qubits, and  unitary operations $U$ acting on a single qubit
\begin{equation}\label{eq:unitary-invariance} \|\sigma - \rho\|_{W_1} =\| U \sigma U^* -  U \rho U^*\|_{W_1}.\end{equation}

The quantum Wasserstein distance of order $1$ can be compared with the trace distance, via the following analogue of \eqref{eq:comparison-TV-W1}:
\begin{equation}\label{eq:comparison-quantum-tr-w1} D_{\tr}(\sigma, \rho) \le \| \sigma - \rho\|_{W_1} \le n D_{\tr} (\sigma, \rho).\end{equation}
To see why this holds,  consider any representation as in \eqref{eq:quantum-w1}
\begin{equation} \sigma - \rho = \sum_{i=1}^n c_i( \sigma^{(i)} -\rho^{(i)}),\end{equation}
take the $1$-Schatten norm on both sides, and use the triangle inequality:
\begin{equation} \tr[ | \sigma - \rho| ] \le \sum_{i=1}^n c_i \tr[ | \sigma^{(i)} -\rho^{(i)} | ] \le 2 \sum_{i=1}^n c_i.\end{equation}
Minimization upon the representations gives the first inequality in \eqref{eq:comparison-quantum-tr-w1}.

To prove the second inequality, assume first that $D_{\tr}(\sigma, \rho) =1$, and let
\begin{equation} \rho^{(i)} = \rho_{\cur{1, \ldots, i}} \otimes \sigma_{\cur{i+1, \ldots, n}},  \quad \quad \sigma^{(i)} = \rho_{\cur{1, \ldots, i-1}} \otimes \sigma_{\cur{i, \ldots, n}}.\end{equation}
Then,
\begin{equation} \rho - \sigma = \sum_{i=1}^n \left(\rho^{(i)} - \sigma^{(i)}\right),\end{equation}
hence
\begin{equation} \|\rho - \sigma\|_{W_1} \le n = n D_{\tr}(\sigma, \rho).\end{equation}
In the general case, we use the spectral theorem on the operator $\rho - \sigma$ to write it as
\begin{equation} \rho - \sigma = (\rho' - \sigma') D_{\tr}(\sigma, \rho),\end{equation}
where $D_{\tr}(\sigma', \rho') = 1$. Since the quantum Wasserstein distance of order $1$ is induced by a norm, we have
\begin{equation} \| \rho - \sigma \|_{W_1} =  D_{\tr}(\sigma, \rho) \| \rho' - \sigma'\|_{W_1} \le n D_{\tr}(\sigma, \rho). \end{equation}

The quantum Wasserstein distance of order $1$ is well-adapted to the tensor product structure of $(\C^2)^{\otimes n}$ as the next proposition shows (see \cite{de2021quantum-1} for a proof).

\begin{proposition}[tensorization]\label{prop:tens}
Given states $\sigma$, $\rho \in \cS( (\C^2)^{\otimes n})$, it holds, for any $I \subseteq \cur{1, \ldots, n}$,
\begin{equation} \| \sigma - \rho\|_{W_1} \ge \|  \sigma_{I} - \rho_{I}\|_{W_1} + \|\sigma_{\cur{1, \ldots, n} \setminus I} - \rho_{\cur{1, \ldots, n} \setminus I}\|_{W_1}\end{equation}
with equality if
\begin{equation} \sigma =  \sigma_{I} \otimes \sigma_{\cur{1, \ldots, n} \setminus I} \quad \text{and} \quad \rho =  \rho_{I} \otimes \rho_{\cur{1\ldots, n} \setminus I}.\end{equation}
\end{proposition}

For diagonal states with respect to the computational basis, i.e.,
\begin{equation} \sigma = \sum_{x \in \{0,1\}^n} p(x) |x \rangle \langle x |, \quad \quad  \rho = \sum_{x \in \{0,1\}^n} q(x) |x \rangle \langle x |,\end{equation}
the quantum Wasserstein distance of order $1$ recovers precisely the Wasserstein distance of order $1$ with respect to the Hamming distance, i.e.,
\begin{equation}\label{eq:w-1-quantum-classical} \| \rho - \sigma\|_{W_1} = W_1( p, q).\end{equation}
Indeed, by the tensorization property for product states (Proposition \ref{prop:tens}), we have that for every $x, y \in \cur{0,1}^n$,
\begin{equation} \| |x \rangle \langle x | - |y \rangle \langle y | \|_{W_1} =  \sum_{i=1}^n 1_{\{x_i \neq y_i\}} = |x-y|_H.\end{equation}
Given any classical coupling $\pi \in \cop(p,q)$, between $p$ and $q$, the triangle inequality yields
\begin{equation}   \| \rho - \sigma\|_{W_1} = \left\| \sum_{x,y}  \pi(x,y) \left( |x \rangle \langle x | - |y \rangle \langle y |\right) \right\|_{W_1} \le  \sum_{x,y} \pi(x,y) | x-y|_H,\end{equation}
yielding that the quantum Wasserstein distance of order $1$ is cheaper. To argue that it is not strictly cheaper, consider any representation as in \eqref{eq:quantum-w1},
\begin{equation} \sigma - \rho = \sum_{i=1}^n c_i ( \sigma^{(i)} - \rho^{(i)}),\end{equation}
and discard the off-diagonal terms of each $\rho^{(i)}$ and $\sigma^{(i)}$, thus defining classical probability distributions $q^{(i)}$, $p^{(i)}$ over $\cur{0,1}^n$ such that
\begin{equation} p-q = \sum_{i=1}^n c_i (p^{(i)} - q^{(i)}).\end{equation}
 After discarding the $i$-th bit, the $n-1$ marginals of $p^{(i)}$ and $q^{(i)}$  coincide (because $\tr_i[\sigma^{(i)}] = \tr_i[ \rho^{(i)}]$), hence
\begin{equation} W_1(p^{(i)}, q^{(i)}) \le 1.\end{equation}
Since also the classical Wasserstein distance of order $1$ is induced by a norm, by triangle inequality we conclude that
\begin{equation} W_1(p,q) \le \sum_{i=1}^n c_i W_1( p^{(i)},  q^{(i)} ) \le \sum_{i=1}^n c_i,\end{equation}
hence equality holds in \eqref{eq:w-1-quantum-classical}.

\subsection{Distance and channels}

Consider a quantum channel $\Phi = \Phi^{(k)}\otimes \I_{\cL( (\C^{2})^{\otimes (n-k)})}$ on a system of $n$ qubits  but acting only on a subset of $k$ qubits (without loss of generality, the first $k$ qubits). Then, we argue that
\begin{equation} \label{eq:distance-rho-phi-rho} \| \rho - \Phi(\rho) \|_{W_1} \le  2 k.\end{equation}
Indeed, letting $\sigma = \Phi(\rho)$, we have that
\begin{equation}\label{eq:trace-1-k-rho-sigma} \tr_{1 \ldots k} [\rho] = \tr_{1 \ldots k}[\sigma].\end{equation}
We claim that the condition \eqref{eq:trace-1-k-rho-sigma} implies
\begin{equation}\label{eq:distance-w1-k} \|  \rho - \sigma \|_{W_1} \le 2 k.\end{equation}
Indeed, let us define, for $i=1,\ldots, k$,
\begin{equation} \sigma^{ \ge i} = 2^{-i} \I_{1 \ldots i} \otimes \tr_{1 \ldots i}[\sigma], \quad \rho^{ \ge i} = 2^{-i} \I_{1 \ldots i} \otimes \tr_{1 \ldots i}[\rho],\end{equation}
where $\I_{1 \ldots i}$ denotes the identity on the composite system of the first $i$ qubits. Then,
\begin{equation} \tr_i[ \sigma^{ \ge {i-1} }] \otimes \I_i /2  = \sigma^{ \ge i}, \quad \tr_i[ \rho^{ \ge {i-1} }] \otimes \I_i /2  = \rho^{ \ge i}\end{equation}
and for $i= k$, using \eqref{eq:trace-1-k-rho-sigma},
\begin{equation}
 \rho^{\ge k} = \sigma^{\ge k}.
\end{equation}
Summing telescopically, we obtain
\begin{align} \rho - \sigma &= \sum_{i=1}^k \left(( \rho^{ \ge i-1} - \rho^{\ge  i}) - ( \sigma^{ \ge i-1} - \sigma^{\ge  i})\right)\nonumber\\
& = 2 \sum_{i=1}^k \left(\frac 1 2 ( \sigma^{ \ge i}  + \rho^{\ge  i-1} ) - \frac 1 2( \rho^{ \ge i}  + \sigma^{\ge  i-1} )\right).\end{align}
Since
\begin{equation} \tr_i[  \sigma^{ \ge i}  + \rho^{\ge  i-1}  ] = \tr_i[  \rho^{ \ge i}  + \sigma^{\ge  i-1}  ],\end{equation}
we have
\begin{equation} \left\|  \frac 1 2 ( \sigma^{ \ge i}  + \rho^{\ge  i-1} ) - \frac 1 2( \rho^{ \ge i}  + \sigma^{\ge  i-1} ) \right\|_{W_1} \le 1\end{equation}
hence \eqref{eq:distance-w1-k} by triangle inequality.

It is also important for applications to understand how much the quantum Wasserstein distance of order $1$ expands under the action of a quantum channel $\Phi$ acting on systems of qubits (in fact, it does not even need to be the same number of qubits). We introduce the $W_1$-contraction coefficient of $\Phi$ as the quantity
\begin{equation}
\left\|\Phi\right\|_{W_1\to W_1} = \max_{\rho\neq\sigma} \frac{\left\|\Phi(\rho)-\Phi(\sigma)\right\|_{W_1}}{\left\|\rho-\sigma\right\|_{W_1}}.
\end{equation}
Simple examples show that it is not true that $\left\|\Phi\right\|_{W_1\to W_1} \le 1$ in general -- indeed, in the classical case the analogous quantity is related to Ollivier’s coarse Ricci curvature \cite{ollivier2007ricci}. For many families of channels this can be efficiently computed or estimated: see \cite{de2021quantum-1} for examples.

\subsection{Lipschitz observables}

We next search for a dual formulation
\begin{equation}\label{eq:wass-duality}
 \| \rho - \sigma\|_{W_1} = \max_{A} \left\{ \tr[ A (\rho - \sigma) ] \, : \, \| A \|_L \le 1\right\},\end{equation}
where maximization is among observables $A \in \cO( (\C^2)^{\otimes n})$ and $\| \cdot \|_L$ denotes a suitable notion of quantum Lipschitz constant of $A$. Recalling that the Wasserstein distance of order $1$ is induced by a norm, it follows that $\| \cdot \|_L$ must be the dual norm:
\begin{equation}\label{eq:abstract-lipschitz} \| A \|_L = \max_{\rho,\sigma} \left\{   \tr[ A ( \rho- \sigma) ] \, : \, \| \rho - \sigma \|_{W_1} \le 1 \right\}.
\end{equation}
If  we take this as a definition, then duality  \eqref{eq:wass-duality} is straightforward. However, this may seem too abstract. Interestingly, we have instead the following ``operational''  definition of the quantum Lipschitz constant:
\begin{equation}\label{eq:operational-lipschitz}
\left\| A \right\|_L = 2\max_{i=1,\ldots, n}\min_{A^{(i)}} \left\|A - \I_i\otimes A^{(i)}\right\|_\infty\,,
\end{equation}
where $A^{(i)}$ runs among the observables that do not act on the $i$-th qubit system, and $\| \cdot \|_\infty$ denotes the operator norm.

To prove the equivalence between \eqref{eq:abstract-lipschitz} and \eqref{eq:operational-lipschitz}, write temporarily
 \begin{equation} \| A \|_{L}'  = 2\max_{i=1, \ldots, n}\min_{A^{(i)}} \left\| A - \I_i\otimes A^{(i)}\right\|_\infty.\end{equation}
 We prove first that $\| A \|_{L} \le \| A\|_{L}'$. Given states $\rho^{(i)}$, $\sigma^{(i)}$ such that $\tr_i[\rho^{(i)}] = \tr_i[\sigma^{(i)}]$, then
\begin{equation} \tr[ \I_i\otimes A^{(i)}( \rho^{(i)} - \sigma^{(i)})] =\tr[ A^{(i)} \tr_i[ \rho^{(i)} - \sigma^{(i)} ] ]  = 0.\end{equation}
 It follows that, given states $\sigma$, $\rho$ and any representation as in \eqref{eq:quantum-w1}, i.e.,
 \begin{equation} \rho - \sigma = \sum_{i=1}^n c_i ( \rho^{(i)} - \sigma^{(i)}),\end{equation}
  we have the identity
\begin{equation} \tr[ A (\rho - \sigma) ] = \sum_{i=1}^n c_i \tr [  A ( \rho^{(i)} - \sigma^{(i)})]  =\sum_{i=1}^n c_i \tr [  (A  - \I_i\otimes A^{(i)})( \rho^{(i)} - \sigma^{(i)})]. \end{equation}
Using the inequality $\tr[BC] \le \| B\|_\infty \| C\|_1$, we bound from above
\begin{equation} \tr[ A (\rho - \sigma) ] \le \sum_{i=1}^n c_i \| A\|_{L}'.\end{equation}
 Assuming $\| \rho - \sigma\|_{W_1} \le 1$, we obtain the inequality $\| A \|_{L} \le \| A\|_{L}'$. To prove the converse we use the general fact that, for an observable $A \in \cO(\cH)$,
\begin{equation}2 \min_{c \in \mathbb{R}} \| A - c \I_{\cH}\|_\infty = \max_{\rho, \sigma \in \cS(\cH)} \tr[ A (\rho - \sigma)],\end{equation}
whose proof (via the spectral theorem) is left as an exercise to the reader. For $i =1, \ldots, n$ we write
\begin{equation}\begin{split} \min_{A^{(i)}} \| A- \I_i \otimes A^{(i)} \|_\infty & =  \min_{A^{(i)}} \min_{c \in \mathbb{R}} \| (A- \I_i \otimes A^{(i)}) - c \I\|_\infty \\
& =  \frac{1}{2}\min_{A^{(i)}} \max_{\rho, \sigma}  \tr[ ( A -  \I_i\otimes A^{(i)})( \rho - \sigma)] \\
& = \frac{1}{2}\max_{\rho, \sigma}\min_{A^{(i)}}   \tr[ ( A -  \I_i\otimes A^{(i)})( \rho - \sigma)],\end{split}\end{equation}
where in the last step we exchanged minimization and maximization by Fenchel duality theorem.
Next, we reduce maximization to $\rho$, $\sigma$ such that $\tr_i[\rho] = \tr_i[\sigma]$, otherwise one can choose a sequence $A^{(i)}_n$ so that
\begin{equation} \tr[\I_i\otimes A^{(i)}_n( \rho - \sigma)] = \tr[ A^{(i)}_n \tr_i[ \rho - \sigma] ] \to  - \infty. \end{equation}
Therefore, $\tr[\I_i\otimes A^{(i)}( \rho - \sigma)] = 0$, and
\begin{equation} \min_{A^{(i)}} \| A - \I_i \otimes A^{(i)} \|_\infty \le  \max_{\rho, \sigma} \tr[ A (\rho - \sigma)] \le \| A \|_L \max_{\rho, \sigma} \| \rho - \sigma\|_{W_1}.\end{equation}
Finally, we use that  $\| \rho - \sigma\|_{W_1} \le 1$, since $\tr_i[\rho] = \tr_i[\sigma]$.

\begin{remark}
It always holds $\| A \|_{L} \le 2 \min_{c \in \mathbb{R}}\| A  - c \I\|_\infty \le 2 \|A\|_\infty$. Moreover, if
\begin{equation} A = \sum_{I\subseteq \cur{1,\ldots, n}}A_I\end{equation}
 is a sum of local operators, i.e., $A_{I}$ acts only on qubits in the subset $I$, then
\begin{equation}
\left\| A \right\|_L \le 2\max_{i=1, \ldots, n}\left\|\sum_{i \in I}A_I\right\|_\infty,
\end{equation}
simply by taking for $i=1, \ldots, n$,
\begin{equation} \I_i \otimes A^{(i)} = \sum_{I\not \ni i} A_I .\end{equation}
\end{remark}

\subsection{Gaussian concentration inequalities}

It is not difficult to argue that, for a Lipschitz observable $A$,  all the eigenvalues $\lambda$ must belong to the interval
\begin{equation} \tr[A]/2^n - n \|A\|_{L}\le \lambda \le \tr[A]/2^n + n \|A \|_{L}.\end{equation}
The above bound however is not very useful, and in fact most eigenvalues belong to a much smaller interval,  of length  $\approx \sqrt{n} \| A \|_{L}$. This is the content of the following result.

\begin{proposition}[concentration inequality]
Given $A \in \cO( (\C^2)^{\otimes n} )$, for every $\delta>0$, it holds
\begin{equation}\label{eq:concentration}
\dim\left( A \ge \left(\tr[A]/2^n+\delta\sqrt{n}\left\|A \right\|_L/2\right) \I \right) \le 2^n\,\exp(-\delta^2/2).
\end{equation}
\end{proposition}

Since $\dim( (\mathbb{C}^2)^{\otimes n} ) = 2^n$, it yields that the relative distribution of eigenvalues is (roughly) concentrated as a Gaussian law, with mean  $\tr[A]/2^n$ and standard deviation $\sqrt{n} \|A \|_L /2$.

A classical tool to establish concentration results are transport-entropy inequalities of the form
\begin{equation}\label{eq:transport-entropy} W_1( \sigma, \rho) \le \sqrt{  \frac{ n}{2}\,D_{KL}( \rho \| \sigma) }\,.\end{equation}
for probabilities $\sigma$, $\rho$, and $D_{KL}( \rho \| \sigma)$ is the relative entropy \eqref{eq:kl}. Relevant cases in the discrete setting, due to K.\ Marton \cite{marton1996bounding}, include the case where $q$ is a product distribution or more generally a Markov chain, under mild assumptions. Following this route, we establish a quantum analogue of \eqref{eq:transport-entropy} for product states on qubits. For an extension to non-product states, see \cite{de2022quantum}.

\begin{proposition}[Quantum Marton's inequality]
For any $\rho, \sigma \in  \mathcal{S}((\mathbb{C}^2)^{\otimes n})$, with
$$ \sigma = \sigma_1 \otimes \ldots \otimes \sigma_n$$
product state, the following inequality holds:
\begin{equation}\label{eq:quantum-marton}
\| \rho - \sigma\|_{W_1} \le \sqrt{  \frac{ n}{2}\,S( \rho \| \sigma) }.
\end{equation}
%where the quantum relative entropy is defined as
%\begin{equation} S(\rho \| \sigma) = \tr[ \rho\left(  \ln(\rho) - \ln(\sigma) \right)].\end{equation}
\end{proposition}

To prove \eqref{eq:quantum-marton}, we argue first in the case $n=1$, so that the quantum Wasserstein distance of order $1$ coincides with the trace distance, and the inequality becomes
\begin{equation}\label{eq:pinsker-quantum}
 D_{\tr}(\sigma, \rho) \le \sqrt{  \frac 1 2 S( \rho \| \sigma) }.
\end{equation}
This inequality is well-known in the literature as the quantum analogue of the classical Pinsker's inequality. Its proof is quite straightforward from the classical Pinsker's inequality: given states $\rho$, $\sigma$, using spectral calculus on $\rho-\sigma$, we introduce the orthogonal projectors $\Pi_+ = \I_{ \{\rho - \sigma \ge 0\}}$, $\Pi_- = \I_{\{ \rho - \sigma <0\}}$ and probabilities on a two-point space $\{ -, +\}$
$$ r_{\pm} = \tr[\Pi_\pm \rho], \quad s_{\pm}  = \tr[ \Pi_\pm \sigma].$$
So that
\begin{equation} 2 D_{\tr}(\sigma, \rho) = |r_+ - s_+| + |r_-  - s_-| = \| r - s \|_1 \le \sqrt{ 2 S(r \| s)},\end{equation}
having used the classical Pinsker's inequality. Since $\{ \Pi_+, \Pi_-\}$ gives a measurement (often called Helstrom measurement), the monotonicity of the quantum relative entropy \eqref{eq:entropy-data-processing} yields that
$$ S( r \| s ) \le S( \rho \| \sigma), $$
hence \eqref{eq:pinsker-quantum}. For the general case, write
\begin{equation} \rho - \sigma = \sum_{i=1}^n \left( \rho_{1\ldots i} \otimes \sigma_{i+1 \ldots n} - \rho_{1\ldots i-1} \otimes \sigma_{i \ldots n}\right),\end{equation}
and apply Pinsker's inequality, for every $i=1, \ldots, n$:
\begin{equation} \left\| \rho_{1\ldots i} \otimes \sigma_{i+1 \ldots n} - \rho_{1\ldots i-1} \otimes \sigma_{i \ldots n}\right\|_1 \le \sqrt{2\, S\left( \rho_{1\ldots i} \otimes \sigma_{i+1 \ldots n} \| \rho_{1\ldots i-1} \otimes \sigma_{i \ldots n}\right)}.
\end{equation}
Summing upon $i$ and using concavity of the square root,
\begin{equation} \| \rho - \sigma\|_{W_1} \le \sqrt{\frac{ n }{2} \sum_{i=1}^n S\left( \rho_{1\ldots i} \otimes \sigma_{i+1 \ldots n} \| \rho_{1\ldots i-1} \otimes \sigma_{i \ldots n}\right) }.
\end{equation}
Finally, we conclude using the identity
\begin{equation}
\begin{split}
 S\left( \rho_{1\ldots i} \otimes \sigma_{i+1 \ldots n} \| \rho_{1\ldots i-1} \otimes \sigma_{i \ldots n}\right)  & =  S\left( \rho_{1\ldots i} \| \rho_{1\ldots i-1} \otimes \sigma_{i}\right) \\
 & = -S(\rho_{1\ldots i} ) +S(\rho_{1\ldots i-1})- \mathrm{tr}\left[ \rho_i \ln \sigma_i\right]
\end{split}\end{equation}
and a telescopic summation.

Starting from \eqref{eq:quantum-marton}, we now deduce the Gaussian concentration for quantum Lipschitz observables. Let $A \in \cO((\C^2)^{\otimes n})$ and for simplicity assume that $\tr[A] = 0$ and $\| A \|_L \le 1$. We prove that
\begin{equation} \tr[ e^{tA} ] \le 2^n \exp( n t^2 /8), \quad \text{for every $t\in \R$,}\end{equation}
so that concentration follows by spectral calculus and Markov inequality. For simplicity, assume that $t>0$ (otherwise argue with $-A$, $-t$ instead of $A$ and $t$). We choose in \eqref{eq:quantum-marton} the state $\sigma = \I/2^n$ which is a product state (also called the maximally mixed state). Duality for the quantum Wasserstein distance of order $1$ and and Marton's inequality yield, for any state $\rho$,
\begin{equation} t \, \tr[ A  \rho ] = t\, \tr[ A (\rho - \sigma)] \le t\|\rho - \sigma\|_{W_1} \le t \sqrt{ \frac{n}{2} S(\rho\| \sigma) } \le \frac{n t^2}{8}+  S( \rho \| \sigma).\end{equation}
We then choose  $\rho = e^{tA}/ \tr[e^{tA}]$ (a Gibbs state) so that $\ln \rho = t A - \ln \tr[e^{tA}]$ and
\begin{equation} S(\rho \| \sigma)  = \tr [ \rho \ln \rho] + \ln 2^n = t \tr[ A  \rho ] - \ln \tr[e^{tA}] + \ln 2^n.\end{equation}
Rearranging the terms we conclude that
\begin{equation} \ln \tr[e^{tA}] \le  \frac{n t^2}{8} + \ln 2^n,\end{equation}
that is equivalent to \eqref{eq:concentration}.

\subsection{Continuity of entropy}
As a second application of the quantum Wasserstein distance of order $1$, we discuss a modulus of continuity for the quantum entropy. In classical information theory, Shannon's entropy
\begin{equation}
S(\rho) = -\sum_{x \in \cX} \rho(x) \ln \rho(x)
\end{equation} quantifies the amount the information contained in a classical probability distribution $(\rho(x))_{x \in \cX}$. Usually the logarithm is in base $2$ (entropy is measured in bits of information), but for simplicity we use the natural basis here. % natural basis (entropy is measured in \emph{nats}) also $\log_2$ (\emph{bits}) is common.
The notation $S(X) = S(\rho)$ is also quite common, where $X$ is a random variable with values in $\cX$ and law $\rho$. For example, the entropy of a random variable $X$ with uniform distribution over $d$ values is $\ln d$, hence if $d=2^n$, we have that $S(X) = n \ln 2$.

For a state $\rho \in \cS(\cH)$ on a quantum system $\cH$, the quantum analogue of Shannon's entropy is defined as
\begin{equation} S(\rho) = - \tr[ \rho \ln \rho],\end{equation}
and was first introduced by von Neumann. Interestingly, this definition was historically prior to Shannon's work -- but this should not be surprising since entropy had already been considered in thermodynamics and statistical physics. As its classical counterpart, the quantum entropy plays a fundamental role in information theory.

Clearly, the entropy is continuous as a function of its argument (as long as we consider finite sets or finite-dimensional quantum systems), however a precise modulus of continuity would be quite useful in applications.  Polyanski and Wu  \cite{polyanskiy2016wasserstein} proved the following explicit bound, for probability distributions $\sigma$, $\rho$ on the discrete cube $\cX = \{0,1\}^n$:
\begin{equation}\label{eq:poly-wu} | S(\sigma) -S(\rho) | \le n h_2\left(\frac{W_1(\sigma,\rho)}{n}\right).\end{equation}
where $h_2(x) =  - (1-x) \ln(1-x) - x \ln x$ the so-called binary entropy function, i.e., Shannon's entropy of a Bernoulli distribution with parameter $x \in [0,1]$. To grasp the relevance of this result, assume that $\sigma$ is uniform over the $2^n$ values, so that $S(\sigma) \approx n$ and $\rho$ is close to $\sigma$, i.e., \ $W_1(\sigma,\rho)$ is much smaller than $n$. Then, $h_2(  W_1(\sigma,\rho)/n )$ is also small, hence $S(\rho)$ is also of order $n$.

In the  quantum setting,  given states  $\rho$, $\sigma \in \mathcal{S}(\cH)$ on a quantum system $\cH$, Fannes \cite{fannes1973continuity} and Audenaert  \cite{audenaert2007sharp} proved  that
\begin{equation} |S(\rho) - S(\sigma)| \le h_2\left( D_{\tr}(\sigma, \rho) \right) + D_{\tr}(\sigma, \rho) \ln ( \dim(\cH)-1).
\end{equation}
When compared to \eqref{eq:poly-wu}, we see that if $\dim(\cH) = 2^{n}$ as in the case of $n$-qubit systems, the inequality becomes much less effective, because of the second term in the right hand side grows linearly as $n$ grows. Moreover, recalling \eqref{eq:comparison-quantum-tr-w1}, it would be better to replace the trace distance with the quantum Wasserstein distance of order $1$, divided by $n$.

This is precisely what was obtained in \cite{de2021quantum-1} and later improved in \cite{de2022wasserstein}: for $\rho$, $\sigma \in \mathcal{S}((\mathbb{C}^2)^{\otimes n})$, it holds
\begin{equation} |S(\rho) - S(\sigma)| \le n h_2 \left( \frac{ \| \rho - \sigma\|_{W_1}}{n} \right) + \| \rho - \sigma\|_{W_1} \ln (3).\end{equation}
Dividing both sides by $n$, one obtains an inequality that can be even extended to the $n =\infty$ case, as investigated in \cite{de2022wasserstein}.

\section{Conclusion}

We presented two recent approaches to the theory of optimal transport for quantum systems.  As described in Section~\ref{sec:overview}, these are not the only ones, and several alternatives have been explored, and possibly other ones will be introduced, since this research field has become quite active, see \cite{toth2022quantum, feliciangeli2021non, bistron2022monotonicity, duvenhage2020quadratic, duvenhage2022quantum} for even more recent and interesting variants. This variety is indeed a resource for the field, since a particular distance may be better suited for the application one has in mind, from quantum computing, machine learning and quantum state tomography \cite{chakrabarti2019quantum,rouze2021learning,onorati2023efficient, kiani2022learning, de2023limitations,hirche2022quantum}, to the study of dynamics and convergence to equilibria of systems \cite{capel2020modified, gao2021ricci, gao2022complete}. From a purely mathematical perspective, we are sure that future investigations on the links between these formulations will provide new insights on the geometrical properties of quantum systems.

% Uncomment these two lines and comment the third if you wish to use bibtex instead of biblatex

%\bibliographystyle{amsalpha}
%\bibliography{budapest-notes-trevisan-depalma.bib}
\printbibliography

\end{document}